\documentclass[letterpaper, 10pt, conference]{cssconf}   

\IEEEoverridecommandlockouts                              
\usepackage{enumitem}
\usepackage{mleftright}
\usepackage{enumitem}
\usepackage{mathematix}
\usepackage{hyperref}
\hypersetup{%
   breaklinks,
   hidelinks,
   hypertexnames	= true,
   plainpages		= false,
}
\usepackage{cleveref}

\usepackage{booktabs,array,arydshln}
\newcolumntype{A}{ >{$} r <{$} @{} >{${}} l <{$} } 

\newlist{theoremcontents}{enumerate}{1}
\setlist[theoremcontents,1]{%
  label=(\bfseries{\roman*}),
  itemindent=0pt,
  wide=0.5\parindent,
  itemsep=0.0pt,
  parsep=0.0pt,
  listparindent=-10pt,%
  afterlabel={{.\nobreakspace}}
}

\newlist{inlinelist}{enumerate*}{1}
\setlist*[inlinelist,1]{%
  label=(\roman*),
}

\newlist{inlinelistalph}{enumerate*}{1}
\setlist*[inlinelistalph,1]{%
  label=(\alph*),
}

\makeatletter
\def\adl@drawiv#1#2#3{%
        \hskip.5\tabcolsep
        \xleaders#3{#2.5\@tempdimb #1{1}#2.5\@tempdimb}%
                #2\z@ plus1fil minus1fil\relax
        \hskip.5\tabcolsep}
\newcommand{\cdashlinelr}[1]{%
  \noalign{\vskip\aboverulesep
           \global\let\@dashdrawstore\adl@draw
           \global\let\adl@draw\adl@drawiv}
  \cdashline{#1}
  \noalign{\global\let\adl@draw\@dashdrawstore
           \vskip\belowrulesep}}
\makeatother

\declarecommand{\refta}{\ref{the:main:a}}
\declarecommand{\reftanl}{\ref*{the:main:a}}
\declarecommand{\reftb}{\ref{the:main:b}}
\declarecommand{\reftbnl}{\ref*{the:main:b}}
\declarecommand{\reftc}{\ref{the:main:c}}
\declarecommand{\reftcnl}{\ref*{the:main:c}}

\declarecommand{\reflema}{\hyperlink{lem:ineq:a}{
    \ref*{lem:ineq}\emph{(i)}}}
\declarecommand{\reflemb}{\hyperlink{lem:ineq:b}{
    \ref*{lem:ineq}\emph{(ii)}}}
\declarecommand{\reflemc}{\hyperlink{lem:ineq:c}{
    \ref*{lem:ineq}\emph{(iii)}}}

\declarecommand{\ie}{\emph{i.e.,}}
\declarecommand{\eg}{\emph{e.g.,}}
\declarecommand{\iff}{iff}
\declarecommand{\iid}{\emph{i.i.d.}}
\declarecommand{\mss}{\emph{m.s.s.}}
\declarecommand{\emss}{\emph{e.m.s.s.}}
\declarecommand{\sdp}{SDP}

\declarecommand{\lschauderf}{{Brouwer Fixed-Point Theorem}}
\declarecommand{\lschauder}{{Brouwer Fixed-Point Theorem}}
\declarecommand{\lschauderfirst}{\emph{Brouwer Fixed-Point Theorem}}
\declarecommand{\lbanachf}{{Banach Fixed-Point Theorem}}

\declarecommand{\posN}{{\N_+}}
\declarecommand{\pinv}{\dagger}

\declarecommand{\xinit}{x_0}
\declarecommand{\nsample}{N}

\declarecommand{\Am}{\mathbf{A}}
\declarecommand{\Bm}{\mathbf{B}}
\declarecommand{\Lm}{\mathbf{A}_\star}
\declarecommand{\At}{\trans{\mathbf{A}}}
\declarecommand{\Bt}{\trans{\mathbf{B}}}
\declarecommand{\Lt}{\trans{\mathbf{A}_\star}}

\declarecommand{\Sig}{\Sigma_0}
\declarecommand{\Sigh}{\hat{\Sigma}_0}
\declarecommand{\Sigd}{\Delta \Sig}

\declarecommand{\oPe}{\Sig \otimes P^\star}
\declarecommand{\oPde}{\Sig \otimes \Delta P}
\declarecommand{\oPhe}{\Sig \otimes \hat{P}}
\declarecommand{\oPhde}{\Delta \Sig \otimes \hat{P}}
\declarecommand{\oP}{{P_\otimes^\star}}
\declarecommand{\oPd}{\Delta P_\otimes}
\declarecommand{\oPh}{\hat{P}_\otimes}
\declarecommand{\oPhd}{\Delta \hat{P}_\otimes}

\declarecommand{\Ym}{Y}
\declarecommand{\Em}{\Delta Y}
\declarecommand{\Yh}{\hat{Y}}
\declarecommand{\Eh}{\Delta \hat{Y}}

\declarecommand{\Yme}{I + \Sm (\oPe)}
\declarecommand{\Eme}{\Sm (\oPde)}
\declarecommand{\Yhe}{I + \Sm (\oPhe)}
\declarecommand{\Ehe}{\Sm (\oPhde)}

\declarecommand{\Sm}{\mathbf{S}}

\declarecommand{\nm}{\bar{n}}

\declarecommand{\covarinf}{X_\infty^\star}
\declarecommand{\covarinfh}{\hat{X}_\infty}
\declarecommand{\covarinfd}{\Delta X_\infty}
\declarecommand{\Lhash}{{(\mathbf{A} + \mathbf{B} \hat{K})^\#}}
\declarecommand{\Lhasht}{\trans{{(\mathbf{A} + \mathbf{B} \hat{K})^\#}}}

\declarecommand{\cost}[1]{J_{K^\star}(#1)}
\declarecommand{\costact}[1]{J_{\hat{K}}(#1)}

\declarecommand{\lbsig}{\bm{\ubar{\alpha}_{\Sigma}}}
\declarecommand{\ubsig}{\bm{\bar{\alpha}_{\Sigma}}}
\declarecommand{\sigm}{\bm{\alpha_{\Sigma}}}
\declarecommand{\sigl}{\bm{\beta_{\Sigma}}}
\declarecommand{\tsig}{\bm{t_{\Sigma}}}

\declarecommand{\epsP}{\bm{\epsilon_P}}
\declarecommand{\epsK}{\bm{\epsilon_K}}
\declarecommand{\epsX}{\bm{\epsilon_X}}

\DeclareMathOperator{\subG}{subG}

\declarecommand{\Symm}{\mathbb S}
\declarecommand{\Symmpd}{\mathbb S_{++}}
\declarecommand{\Symmpsd}{\mathbb S_{+}}

\declarecommand{\op}[1]{\mathcal{#1}}
\declarecommand{\opmat}[1]{#1}

\declarecommand{\ric}{\op{R}}
\declarecommand{\lyap}{\op{L}}
\declarecommand{\lyapcl}{\op{L}_\star}
\declarecommand{\lyaph}{{\op{L}^\#}}
\declarecommand{\lyaphsa}{{\op{L}^\#_\star}}
\declarecommand{\ricf}{\lyap}
\declarecommand{\ricz}{\ric_0}
\declarecommand{\ricd}{\ric_\Delta}
\declarecommand{\ricfp}{\Phi}

\declarecommand{\lyapfp}{\Phi_{\lyap}}

\declarecommand{\F}{\op{F}}
\declarecommand{\G}{\op{G}}
\declarecommand{\H}{\op{H}}

\declarecommand{\Fo}{\op{F}_\star}
\declarecommand{\Ho}{\op{H}_\star}

\declarecommand{\Pc}{P_I}

\declarecommand{\Fc}{{\op{F}^\#}}
\declarecommand{\Hc}{{\op{H}^\#}}
\declarecommand{\Gc}{{\op{G}^\#}}

\declarecommand{\etaA}{\eta_\F}
\declarecommand{\kappaA}{\kappa_\F}
\declarecommand{\etaB}{\eta_\G}
\declarecommand{\etaBR}{\bar{\eta}_\G}
\declarecommand{\kappaB}{\kappa_\G}
\declarecommand{\etaAB}{\eta_\H}
\declarecommand{\kappaAB}{\kappa_\H}

\declarecommand{\etaL}{\eta_\F^\star}
\declarecommand{\kappaL}{\kappa_\F^\star}
\declarecommand{\etaLB}{\eta_\H^\star}
\declarecommand{\kappaLB}{{\kappa_\H^\star}}

\declarecommand{\etaBc}{\kappa_{\G}^\#}
\declarecommand{\etaLBc}{\kappa_{\H}^\#}
\declarecommand{\kappaBc}{\kappa_\G^\#}

\declarecommand{\etaLyap}{{\kappa_\lyap^\star}}
\declarecommand{\etaLyaph}{{\kappa_{\lyap}^\#}}
\declarecommand{\kappaS}{\kappa_{\ricz}}

\declarecommand{\kappaK}{\kappa_K}
\declarecommand{\etaK}{\eta_K}

\newcommand{\preliminary}[1]{{\leavevmode\color{MidnightBlue} #1}}
\renewcommand{\preliminary}[1]{}
\newcommand{\todo}[3]{\preliminary{{\color{#1} \small [TODO] \textbf{#2} --- #3}}}

\newcommand{\peter}[1]{\todo{MidnightBlue}{Peter}{#1}}

\ifdefined\isarchiv
\newcommand{\archiv}[2]{{\leavevmode\color{SlateBlue2} #1}}
\else
\newcommand{\archiv}[2]{{\leavevmode\color{black!100} #2}}
\fi

\title{\LARGE \bf
Sample Complexity of Data-Driven Stochastic LQR with Multiplicative Uncertainty
}

\author{Peter Coppens and Panagiotis Patrinos
\thanks{P. Coppens and P. Patrinos are with the Department of Electrical
Engineering (ESAT-STADIUS), KU Leuven, Kasteelpark Arenberg
10, 3001 Leuven, Belgium.
        {Email: \tt\footnotesize peter.coppens@kuleuven.be, panos.patrinos@kuleuven.be}}%
\thanks{The work of the first author is supported by: the Research Foundation
Flanders (FWO) PhD grant 11E5520N and research projects G086518N and G086318N;
Research Council KU Leuven C1 project No. C14/18/068; Fonds de la Recherche 
Scientifique – FNRS and the Fonds Wetenschappelijk Onderzoek – Vlaanderen under 
EOS project no 30468160 (SeLMA).}%
}%

\begin{document}

\maketitle
\thispagestyle{empty}
\pagestyle{empty}

\begin{abstract}
This paper studies the sample complexity of the stochastic Linear Quadratic
Regulator when applied to systems with multiplicative noise. We assume that the 
covariance of the noise is unknown and estimate it using 
the sample covariance, which results in suboptimal behaviour. The main 
contribution of this paper is then to bound the suboptimality of the methodology
and prove that it decreases with $1/\nsample$, where $\nsample$ denotes the 
amount of samples. The methodology easily generalizes to the case where the mean 
is unknown and to the distributionally robust case studied in a previous work
of the authors \cite{Coppens2019}. The analysis is mostly based on results from 
matrix function perturbation analysis \cite{Konstantinov1993}.

\end{abstract}

\section{INTRODUCTION}
The field of learning control has recently seen explosive growth, which can
be attributed to the availability of large amounts
of data, creating an incentive for controllers that use the available 
information optimally. A significant amount of this research effort is being 
directed towards the familiar Linear Quadratic Regulation (LQR) problem where
the transition matrices are unknown \cite{Dean_2019, Mania2019, Fazel2018}. 
Most of these developments however are related to deterministic systems. 

Instead this paper takes a different approach, considering systems that 
intrinsically include the uncertainty in the dynamics through stochastic 
disturbances. More specifically we study systems with a time-varying
multiplicative disturbance. These may cover a wide range of system classes like
Linear Parameter Varying (LPV) systems \cite{Wu1996, Byrnes1979} and Linear
Difference Inclusions (LDI) \cite{Boyd1994} or in our case, when the disturbance
varies stochastically, systems with multiplicative noise. Such systems have
already been studied in the context of learning control by using policy 
iteration \cite{Gravell2019} and intrinsically introduce robustness in the 
controller design \cite{Bernstein1986}.

The authors previously developed a control synthesis procedure using the
\emph{distributionally robust approach} that guarantees 
stability with high probability, when the true distribution of the system is not
known. This paper is related to that result and provides a methodology to 
evaluate the performance of the \emph{empirical approach}, where the sample 
mean and covariance are used to produce a controller making it similar to the 
\emph{certainty equivanlent approach} for deterministic LQR. Therefore 
the proofs are similar to the result of Mania et. al. 
\cite{Mania2019}, where the sample complexity of this 
{certainty equivalent approach} is studied.

The main result is then a suboptimality guarantee for the {empirical} 
controller. To produce such a result we make use of Riccati perturbation 
analysis. This paper is, to the authors' knowledge, the first instance of such a 
perturbation analysis being applied to discrete time systems with multiplicative
noise. A Riccati perturbation bound for continuous time systems was already
produced in \cite{Chiang2013}. 

The remainder of this paper is then structured as follows.
Section~\ref{sec:prob} presents the problem statement and the assumptions used
throughout the paper. The main result is then presented in Section~\ref{sec:main} in
the form of three theorems that show how the uncertainty on the covariance propagates 
throughout the controller synthesis. 
The proof of these three components are
then given in the following sections. Section~\ref{sec:pre} lists some
results that are required for the remainder of the derivations as well as 
a way of deriving confidence bounds for the sample covariance. 
Section~\ref{sec:ric} then extends upon the results of 
Konstantinov et. al. \cite{Konstantinov1993} to study the perturbed Riccati
equation. Section~\ref{sec:ctr} uses a result from convex analysis 
to derive a bound for the perturbation of the controller. 
Then Section~\ref{sec:reg} proofs the  main suboptimality 
bound, from which a sufficient condition for mean square stability 
(\mss) of the true system under the empirical controller also follows.
Finally Section~\ref{sec:con} provides a conclusion and suggestions for 
further work.
\subsection{Notation}
Let $\Re$ denote the reals, $\N$ the naturals and $\posN = \N 
\setminus \{0\}$. We use $\Symm^n$ to denote the set of $n$-by-$n$ 
symmetric matrices.
The set of positive (semi)definite matrices is then written as $\Symmpd^n$
$(\Symmpsd^n$). Then, for $P, Q \in \Symm^n$, we write $P \sgt Q$ 
($P \sgeq Q$) to signify that $P - Q \in \Symmpd^n$ ($P - Q \in \Symmpsd^n$). 
We denote by $\otimes$ the Kronecker product, by $A^\pinv$ the pseudoinverse
of some matrix $A$.. 
We assume that all random 
variables are defined on a probability space $(\Omega, \mathscr{F}, \prob)$, 
with $\Omega$ the sample space, $\mathscr{F}$ its associated $\sigma$-algebra 
and $\prob$ the probability measure. Let $y : \Omega \rightarrow \Re^n$ be a 
random vector defined on $(\Omega, \mathscr{F}, \prob)$. With some abuse of 
notation we will write $y \in
\mathbb{R}^n$ to state the dimension of this random vector. Let $\prob_y$ 
denote the distribution of $y$, \ie{} $\prob_y(A) = \prob[y \in A]$, then a 
trajectory $\{y_i\}_{i = 1}^{N}$ of independent and identically distributed 
(\iid) copies of $y$ is defined by the distribution it induces. That is, 
for any $A_0, \ldots, A_N \in \mathscr{F}$ we define $\prob_y(A_0 \times \dots 
\times A_N) \dfn \prob[y_0 \in A_0 \land \dots \land y_N \in A_N] = 
\prod_{i=0}^{N} \prob_y(A_i)$. This definition can be extended to infinite 
trajectories $\{y_i\}_{i\in \N}$ by Kolmogorov's existence theorem 
\cite{Billingsley1995}. We will write the expectation operator as $\E$. 
We denote by $\E\left[ y \mid z \right]$ the conditional 
expectation with respect to $z$. For matrices we will use $\nrm{\cdot}$ to 
denote the spectral norm and $\nrm{\cdot}_F$ to denote the Frobenius norm. 
For a linear matrix operator $\F: \Re^{n\times n} \rightarrow \Re^{n\times n}$ 
we similarly use $\nrm{\F}$ to denote the operator-norm defined as 
$\nrm{\F} \dfn \max_{\nrm{X} \leq 1} \nrm{\F(X)}$.

\section{PROBLEM STATEMENT} \label{sec:prob}
In this section, we describe the problem statement and state the main
result.
\subsection{LQR for systems with multiplicative noise}
This paper considers linear systems with input- and state-multiplicative noise
given by:
\begin{equation} \label{eq:dyn}
   x_{k+1} = A(w_k)x_k + B(w_k)u_k,
\end{equation}
with $A(w) \dfn A_0 + \ssum_{i=1}^{n_w} w^{(i)}A_i$ and 
$B(w) \dfn B_0 + \ssum_{i=1}^{n_w} w^{(i)} B_i$, where at each time $k$,
$x_k \in \Re^{n_x}$ denotes the state, $u_k \in \Re^{n_u}$ the input and
$w_k \in \Re^{n_w}$ an \iid{} copy of a square integrable random vector $w$
distributed according to $\prob_w$. We use $w^{(i)}$ to denote the $i$'th 
element of $w$. We introduce the following shorthands:
$\Am \dfn \trans{\smallmat{
   \trans{A_0} & \trans{A_1} & \ldots & \trans{A_{n_w}}
}}$,
$\Bm \dfn \trans{\smallmat{
   \trans{B_0} & \trans{B_1} & \ldots & \trans{B_{n_w}}
}}$ and define $\Sig = \smallmat{1 & 0 \\ 0 & \Sigma}$, where we assume
that $\E[w] = 0$ and $\E[w \trans{w}] = \Sigma$. The $\#$ operator, when applied
to a matrix, then denotes
the block transpose, \ie{} $\Am^\# \dfn \trans{\smallmat{
   A_0 & A_1 & \ldots & A_{n_w}
}}$

The primary goal is to study solutions of the following
stochastic LQR problem:
\begin{equation} \label{eq:lqr}
   \begin{aligned}
      & \minimize_{u_0, u_1, \ldots} & & \E \left[ \ssum_{k=0}^\infty 
      \trans{x_k} Q x_k + \trans{u}_k R u_k \right] \\
      & \stt & & x_{k+1} = A(w_k) x_k + B(w_k) u_k, \quad k \in \N \\
      &     & & x_0\ \textrm{given}
   \end{aligned}
\end{equation}
where we assume that $Q \sgt 0$ and $R \sgt 0$.\footnote{This assumption
is not strictly necessary, see \cite{Chen1998} for some discussion.
} The solution of \eqref{eq:lqr} will yield
a controller that renders the closed-loop system exponentially mean square stable (\emss)
\cite[Definition 1]{Coppens2019}. Note that for the dynamics in \eqref{eq:dyn} 
\mss{} is equivalent to \emss{} \cite[Theorem 2]{Coppens2019}. Therefore we
will say a system is \mss{} throughout the paper, thereby also implying it is \emss{}.

The solution of \eqref{eq:lqr} is then described by the 
following result \cite[Proposition 3]{Coppens2019}:
\begin{proposition}[LQR control synthesis] \label{prob:lqr}
   Consider a system with dynamics \eqref{eq:dyn} and the associated LQR problem
   \eqref{eq:lqr}. Assuming that \eqref{eq:dyn} is mean square stabilizable, 
   \ie{} there exists a $K$, such that the closed-loop system 
   $x_{k+1} = (A(w_k) + B(w_k)K)x_k$ is \mss{}, then the following statements
   holds.
   \begin{theoremcontents}
      \item The optimal solution of $\eqref{eq:lqr}$ is given by $K^\star = 
      -(R+\G(P^\star))^{-1}\H(P^\star)$, with $P^\star$ the solution of the following
      Riccati equation:
      \begin{align}
         \ric(P^\star, \Sig) &\dfn P^\star - Q - \F(P^\star) \nonumber \\ &+ 
         \trans{\H(P^\star)}(R + \G(P^\star))^{-1} \H(P^\star) = 0, \label{eq:ric}
      \end{align}
      with the linear maps $\F(P)$, $\G(P)$, $\H(P)$ defined in 
      Table~\ref{tab:constants}.
      \item The controller $K^\star$ renders \eqref{eq:dyn} \mss{} in closed-loop.
      \item The optimal cost is given by \\ $\cost{x_0} =
      \E [\ssum_{k=0}^\infty \trans{x_k} (Q + \trans{{K^\star}}RK^\star) x_k] = 
      \trans{x_0} P^\star x_0$. 
   \end{theoremcontents}
\end{proposition}
The goal of this paper is then to consider the effect of misestimation of
$\Sigma$ on the closed-loop cost. 
More specifically we will operate under the following assumption
\begin{assumption} \label{ass:covar}
   Let $\Sigh = \Sig + \Sigd$ be some estimator of $\Sig$,
   using $\nsample$ samples of the random vector $w$. We will assume it
   satisfies the following:
   \begin{equation} \label{eq:confsig}
      \lbsig \Sig \sleq \Sigd \sleq \ubsig \Sig,
   \end{equation}
   where $-1 \leq \lbsig \leq 0 \leq \ubsig = 
   \mathcal{O}\left(1/\sqrt{N}\right)$. 
\end{assumption} 
This assumption is valid with high probability when 
$\Sigh = \smallmat{1 & 0 \\ 0 & \hat{\Sigma}}$
where $\hat{\Sigma} = \ssum_{i=1}^{\nsample} w_i \trans{w_i}$ and under some 
additional assumptions on $w$, which are stated in Section~\ref{sec:pre}. 
It is also applicable for the case where the mean is also unknown
and estimated as the sample mean. The constants $\lbsig$ and $\ubsig$ depend
on $N$, which is made explicit by using bold symbols.

\section{MAIN RESULT} \label{sec:main}
Starting from this assumption we will study the optimal controller
produced by applying \Cref{prob:lqr} for $\Sigma_0$ and $\hat{\Sigma}_0$ which 
we will denote as $K^*$ (\emph{nominal controller}) and $\hat{K}$ 
(\emph{empirical controller}) respectively. The goal is then to quantify
the difference between $\cost{\xinit}$ and $\costact{\xinit}$. To do so we
study how the perturbation on $\Sigma_0$ propagates through the controller 
synthesis in three stages. The first stage is how the solution of the Riccati
equation is perturbed, which is quantified in Theorem~\refta{}.
The second stage is the perturbation of the control gain, quantified in
Theorem~\reftb{}. The final stage is then the suboptimality, quantified in
Theorem~\reftc{}. 

We state these theorems for a system with dynamics \eqref{eq:dyn}, 
with $\E[w] = 0$ and 
$\E[w \trans{w}] = \Sigma$ and $K^\star$ the optimal controller and $P^\star$ 
the solution of \eqref{eq:ric}. Then assume we have some $\Sigh = \Sig + \Sigd$
which satisfies Assumption~\ref{ass:covar} and denote by $\hat{K}$ the 
optimal controller for $\Sigh$ and $\hat{P}$ the solution of \eqref{eq:ric}. 
The constants used in the theorems below are listed in Table~\ref{tab:constants}.

\begin{table*}[t]
   {\footnotesize
   \caption{Overview of system constants}\label{tab:constants}
   \begin{center}
   \begin{tabular}{AAA} 
   \toprule
   \multicolumn{2}{l}{System Matrices:} &
   \Lm &= \Am + \Bm K^\star & \Sm &= \Bm R^{-1} \Bt \\
   \multicolumn{2}{l}{Spectra:} &
   \mu_{P}^\star &= \min \sigma(P^\star) & \mu_{R} &= \min \sigma(R) \\
   \midrule
   \multicolumn{2}{l}{Operators} & \multicolumn{2}{l}{Sensitivities} &
   \multicolumn{2}{l}{Offsets} \\
   \midrule
   \F&:\Symm^{n_x} \to \Symm^{n_x}: P \mapsto \At(\Sig \otimes P)\Am & 
      \kappaA &= \nrm{\F}   & \etaA &= \nrm{\F(P^\star)} \\
   \G&:\Symm^{n_x} \to \Symm^{n_u}: P \mapsto \Bt(\Sig \otimes P)\Bm  &
      \kappaB &= \nrm{\G}   & \etaB &= \nrm{\G(P^\star)} \\
   \H&: \Symm^{n_x} \to \Re^{n_u \times n_x}:  P \mapsto \Bt(\Sig \otimes P)\Am & 
      \kappaAB &= \nrm{\Am}\nrm{\Bm}\nrm{\Sig} & \etaAB&=\nrm{\Am}\nrm{\Bm}\nrm{\Sig}\nrm{P^\star} \\ 
   \cdashlinelr{1-6}
   \Fo&:\Symm^{n_x} \to \Symm^{n_x}: P \mapsto \Lt(\Sig \otimes P)\Lm & 
      \kappaL &= \nrm{\Fo} & \etaL &= \nrm{\Fo(P^\star)} \\
   \Ho&: \Symm^{n_x} \to \Re^{n_u \times n_x}:  P \mapsto \Bt(\Sig \otimes P)\Lm & 
      \kappaLB &= \nrm{\Lm}\nrm{\Bm}\nrm{\Sig}   & \etaLB &= \nrm{\Lm}\nrm{\Bm}\nrm{\Sig}\nrm{P^\star} \\ 
   \lyapcl&: \Symm^{n_x} \to \Symm^{n_x}:  P \mapsto  P - \Lt (\Sig \otimes P) \Lm &
      \etaLyap &= \nrm{\lyapcl^{-1}}^{-1} & \\
   \cdashlinelr{1-6}
   \Gc&:\Symm^{n_u} \to \Symm^{n_x}: X \mapsto \trans{{\Bm^\#}}(\Sig \otimes X)\Bm^\# &
      \kappaBc &= \nrm{\Gc} & \\
   \Hc&: \Re^{n_u \times n_x} \to \Re^{n_x \times n_u}:  X \mapsto \trans{{\Bm^\#}}(\Sig \otimes X)\Lm^\# & 
      \etaLBc &= \nrm{\Am^\#}\nrm{\Bm^\#}\nrm{\Sig} & \\ 
   \lyaphsa&: \Symm^{n_x} \to \Symm^{n_x}:  X \mapsto  X - \trans{{\Lm^\#}}(\Sig \otimes X) \Lm^\# &
      \etaLyaph &= \nrm{\lyaphsa^{-1}}^{-1} & \\
   \midrule
   \multicolumn{2}{l}{Other:} &
   \kappaS &= \nrm{\Lm}^2 \nrm{\Sm} \nrm{\Sig}^2 & \etaBR &= \nrm{\G(P^\star) + R} \\
   \bottomrule
   \end{tabular}
   \end{center}
   }
\end{table*}

\begin{theorem}[Riccati Perturbation] \label{the:main:a}
   The distance between the 
   solutions of the Riccati equations $P^*$ and $\hat{P}$ for covariances
   $\Sig$ and $\Sigh$ respectively is bounded as follows:
   \begin{equation} \label{eq:ricp}
      \nrm{P^\star - \hat{P}} \leq \tfrac{\etaLyap -  \kappaA \sigm
   -\sqrt{(\etaLyap - \kappaA \sigm)^2 - 4 \etaA \kappaS \sigm}}
   {2 \kappaS},
   \end{equation}
   with $\sigm = \max\{|\lbsig|, |\ubsig|\}$. This bound holds as long as the 
   following conditions are satisfied:
   \begin{equation} \label{eq:ricpcond}
      \sigm \leq  
      \tfrac{2 \etaA \kappaS + \kappaA \etaLyap - 2 \sqrt{\etaA \kappaS
      (\etaA\kappaS + \kappaA \etaLyap)}}{\kappaA^2},
   \end{equation}
   and $\sigm$ sufficiently small such that the right side of
   \eqref{eq:ricp} is smaller than $\mu_{P}^\star = \min \sigma(P^\star)$. 
\end{theorem}
\begin{proof}
   See Section~\ref{sec:ric} for the proof.
\end{proof}
\begin{theorem}[Controller Perturbation] \label{the:main:b}
   The distance between the optimal 
      controllers for $\Sig$ and $\Sigh$ is bounded as follows:
      \begin{equation} \label{eq:ctrp}
         \nrm{K^\star - \hat{K}} \leq \tfrac{1}{\mu_R}[(1 + \sigm)\epsP
         \kappaK + \sigm \etaK],
      \end{equation}
      with $\epsP$ the right-hand side of \eqref{eq:ricp}, $\kappaK = \kappaB
      \nrm{K^\star} + \kappaAB$, $\etaK = \etaB \nrm{K^\star} + \etaAB$ and $\mu_R = 
      \min \sigma(R)$. 
\end{theorem}
\begin{proof}
   See Section~\ref{sec:ctr} for the proof.
\end{proof}
\begin{theorem}[Suboptimality] \label{the:main:c}
   The difference between the optimal cost 
   $\cost{x_0}$ and the closed-loop cost achieved when applying $\hat{K}$ 
   to the true system, denoted by $\costact{x_0}$, is bounded as follows:
   \begin{align} 
      &\costact{\xinit} - \cost{\xinit}  \nonumber \\
      & \quad \leq 
      {\nm \epsK^2 \etaBR \nrm{\xinit}^2} \tfrac{1}{
      \etaLyaph}\left[1 +  \tfrac{2 \etaLBc \epsK
      + \kappaBc \epsK^2}
      {\etaLyaph - 2 \etaLBc \epsK 
      - \kappaBc \epsK^2}\right],\label{eq:regret}
   \end{align}
   where $\bar{n} = \min(n_x, n_u)$ and $\epsK$ the right-side
   of \eqref{eq:ctrp}. The bound in \eqref{eq:regret} is valid as long as:
   \begin{equation} \label{eq:stab}
      \epsK < \tfrac{\etaLyaph}{\etaLBc + \sqrt{{\etaLBc}^2 +
      \kappaBc \etaLyaph}}.
   \end{equation}
\end{theorem}
\begin{proof}
   See Section~\ref{sec:reg} for the proof.
\end{proof}

The rate of decrease predicted by these theorems is then given in the Corollary
below.

\begin{corollary}[Suboptimality bound] \label{cor:subopt}
   Let $\Sigh$ satisfy Assumption~\ref{ass:covar} and let $K^\star$ denote the
   nominal controller and $\hat{K}$ the emprical controller. Then
   \begin{equation} \label{eq:regret_order}
      \costact{\xinit} - \cost{\xinit} = \mathcal{O}\left(1/\nsample \right),
   \end{equation} 
   assuming that $N$ is sufficiently large. 
\end{corollary}
\begin{proof}
The proof is quite straightforward. Note that $\sigm = \mathcal{O}(1/\sqrt{N})$
by Assumption~\ref{ass:covar}. Let $\epsP$ denote the right-hand side of
\eqref{eq:ricp}. Evaluating the limit $\lim_{N \rightarrow \infty} \sqrt{N}\epsP $
results in
{
   \allowdisplaybreaks
\begin{align*}
   &\lim_{N \rightarrow \infty} \sqrt{N} \left(\tfrac{\etaLyap -  \kappaA \sigm
      -\sqrt{(\etaLyap - \kappaA \sigm)^2 - 4 \etaA \kappaS \sigm}}
      {2 \kappaS}\right) \\
   &= \lim_{N \rightarrow \infty} \tfrac{
         4 \etaA \kappaS (\sqrt{N} \sigm)}
      {2 \kappaS ((\etaLyap -  \kappaA \sigm
      +\sqrt{(\etaLyap - \kappaA \sigm)^2 - 4 \etaA \kappaS \sigm})} \\
   &= \frac{ 4 \etaA \kappaS c}{4 \kappaS \etaLyap},
\end{align*}
where the final equality follows from the fact that 
$\lim_{N \rightarrow \infty} \sqrt{N} \sigm = c >0$ 
(since $\sigm = \mathcal{O}(1/\sqrt{N})$ by assumption), 
which implies that the limit of the numerator is $4 \etaA \kappaS c$. 
The limit of the denominator meanwhile is $4 \kappaS \etaLyap$, since
$\lim_{N \rightarrow \infty} \sigm = 0$.
The overall limit being some positive constant then directly implies that 
$\epsP = \mathcal{O}(1/\sqrt{N})$.
}

Let $\epsK$ be the right-hand side of \eqref{eq:ctrp} then we can see
that $\epsK = \mathcal{O}(1/\sqrt{N})$, since it
depends linearly on $\sigm \epsP = \mathcal{O}(1/N)$ and $\sigm$. Finally 
\eqref{eq:regret} implies the required result since the factor in square
brackets is $\mathcal{O}(1)$ -- which can be seen by noting that the limit for
$N \rightarrow \infty$ is one -- and since $\epsK^2 = \mathcal{O}(1/N)$. 
\end{proof}
Note that the rate predicted by Corollary~\ref{cor:subopt} 
is the same as the one achieved by the {certainty
equivalent} controller for deterministic LQR \cite{Mania2019}.  

\section{PRELIMINARY RESULTS} \label{sec:pre}
In this section we provide some results that will be used throughout
the remainder of this paper. First we slightly alter a previous result from
high-dimensional statistics that results in a condition on $\Sigh$ as in 
\eqref{eq:confsig}. Second we introduce three lemmas that are related to 
bounding the operator norms of versions of $\F$, $\G$ and $\lyap$. 

\subsection{Concentration inequalities for the sample covariance}
When using the sample-covariance $\hat{\Sigma} = \ssum_{i=0}^M w_i \trans{w_i}$,
with $\left\{ w_i \right\}_{i=0}^M$ \iid{} copies of $w$,
we can find a high confidence bound of the parameters $\lbsig$ and $\ubsig$ 
under the following assumptions:
\begin{assumption} \label{ass:w}
   We assume that \begin{inlinelist}
      \item $w$ is square integrable,
      \item $w_k$ and $w_\ell$ are independent for all $k \neq \ell$,
      \item $\E[w] = 0$
      \item $\E[w \trans{w}] = \Sigma \sgt 0$ and
      \item $\Sigma^{-\nicefrac{1}{2}} w \sim \subG_{n_w}(\sigma^2)$ for
      some $\sigma \geq 1$. 
   \end{inlinelist}
\end{assumption}
Here we follow the definition of a sub-Gaussian random vector (denoted by 
$\subG$) given in \cite[Definition 5]{Coppens2019}. Condition (iv) holds for 
example for gaussian $w$ ($\sigma = 1$) and for $w$ with bounded support
(where $\sigma$ can be estimated from data \cite{Delage2010}). Under these 
assumptions we can prove a slightly altered version of 
\cite[Theorem 8]{Coppens2019}, which is stated as:
\begin{theorem} \label{the:conf}
   Let $w \in \Re^{n_w}$ be a random vector satisfying Assumption~\ref{ass:w}
   and $\hat{\Sigma}$ the sample covariance as defined above. Then with
   probability at least $1-\beta$,
   \begin{equation} \label{eq:confsigh}
      -\tsig \Sigma \leq \hat{\Sigma} - \Sigma \leq \tsig \Sigma,
   \end{equation}
   with $\tsig \dfn \tfrac{\sigma^2}{1-2\epsilon}{\Big (}\textstyle\sqrt{\tfrac
   {32q(\beta, \epsilon, n_w)}{M}} + \tfrac{2q(\beta,\epsilon,n_w)}{M}{\Big )}$,
   $\epsilon \in \interval[open]{0}{\nicefrac{1}{2}}$ chosen freely
   and $q(\beta, \epsilon, n_w) \dfn n_w \log(1 + \nicefrac{1}{\epsilon}) +
   \log(\nicefrac{2}{\beta})$. 
\end{theorem}
\begin{proof}
   The proof is a specialised version of that of \cite[Theorem 8]{Coppens2019} 
   and combines \cite[Lemma A.1.]{Hsu2012b} with the methodology of 
   \cite{Delage2010}. The major difference is that no uncertainty on the mean
   is considered and the difference $\hat{\Sigma} - \Sigma$ is bounded instead
   of simply finding an upper bound for $\hat{\Sigma}$. 
\end{proof}
Note $-\tsig \Sig \sleq \Sigd \sleq \tsig \Sig$ follows directly from
\eqref{eq:confsigh}.

\subsection{Norms of matrix operators}
We will consider bounding norms associated with $\F$ and $\G$ in two
circumstances. The first being where we have some $\Sigd$ that is constrained
by \eqref{eq:confsig}. The second being the case where we have some
$\nrm{\Delta P} \leq \epsilon$. To deal with these two cases we will use the
lemmas given below.
\begin{lemma} \label{lem:bndineq}
   Consider the matrices $\mathbf{A} \in \Re^{n_x n_w \times p}$, 
   $P \in \Symmpsd^{n_x}$, $\Sig \in \Symmpsd^{n_w}$
   and $\Sigd \in \Symm^{n_w}$, where
   $\lbsig \Sig \sleq \Sigd \sleq \ubsig \Sig$. Let $\sigm = \max\{|\lbsig|, |\ubsig|\}$
   We can then state the following bound:
   \begin{align*}
       \nrm{\At(\Sigd \otimes P)\Am}
       \leq \sigm\nrm{\At(\Sig \otimes P)\Am}.
   \end{align*}
\end{lemma}
\begin{proof}
   From $\lbsig \Sig \sleq \Sigd \sleq \ubsig \Sig$, due to the fact that
   the eigenvalues of a kronecker product of two matrices are the products of
   the eigenvalues of the matrices, we have that
   $\lbsig \Sig \otimes P \sleq \Sigd \otimes P \sleq \ubsig \Sig \otimes P$, implying $\lbsig \At(\Sig \otimes P)\Am
   \sleq \At(\Sigd \otimes P)\Am \sleq \ubsig \At(\Sig \otimes P)\Am$, 
   which implies the required result.
\end{proof}

\begin{lemma} \label{lem:bnd2}
   Consider the matrices $\mathbf{A} \in \Re^{n_x n_w \times p}$, 
   $P \in \Symmpsd^{n_x}$ and ${\Sigma} \in \Symmpsd^{n_w}$.
   Suppose $\nrm{P} \leq
   \epsilon$ then:
   \begin{equation*}
       \nrm{\trans{\mathbf{A}}({\Sigma} \otimes P)\mathbf{A}} \leq 
       \epsilon \nrm{\trans{\mathbf{A}}(\Sigma \otimes I)\mathbf{A}}.
   \end{equation*}
\end{lemma}
\begin{proof}
   Note that $\nrm{P} \leq \epsilon$ implies $P \sleq \epsilon I$. Therefore
   we can prove the required result using the same arguments as for the proof
   of Lemma~\ref{lem:bndineq}.
\end{proof}

Using Lemma~\ref{lem:bnd2} we can see that 
$\kappaA = \nrm{\At(\Sig \otimes I)\Am}$, since $\nrm{\F} \dfn
\max_{\nrm{P} \leq 1} \nrm{\At(\Sig \otimes P)\Am}$. 
Analogously we can find $\kappaB$, $\kappaBc$
and $\kappaL$. The lemma is however not applicable to $\kappaAB$, which is why we
define it as $\kappaAB = \nrm{\Am} \nrm{\Bm}$. The same is true for $\kappaLB$
and $\etaLBc$. The applicability of
Lemma~\ref{lem:bndineq} is less direct and will be used in Section~\ref{sec:ric}
and Section~\ref{sec:reg}.
To evaluate $\etaLyap$ and $\etaLyaph$
we use Lemma~\ref{lem:lyapnorm}, which is similar to a result for deterministic
dynamics \cite{Gahinet1990}:
\begin{lemma} \label{lem:lyapnorm}
   Let $\lyapcl(P)$ be an invertible Lyapunov operator as defined 
   in Table~\ref{tab:constants}. Then,
   \begin{equation}
      \nrm{\lyapcl^{-1}} = \nrm{\lyapcl^{-1}(I)}.
   \end{equation}
\end{lemma}
\begin{proof}
   First note that $\nrm{I} = 1$. Therefore $\nrm{\lyapcl^{-1}} \geq
   \nrm{\lyapcl^{-1}(I)}$. To prove $\nrm{\lyapcl^{-1}} \leq
   \nrm{\lyapcl^{-1}(I)}$ note that $\trans{x_0} \lyapcl^{-1}(Q) x_0 = 
   \E[\ssum_{k=0}^\infty \trans{x_k} Q x_k]$, with $x_{k+1} = A(w_k)x_k$
   \cite{Morozan1983}. We can write this as $\tr{(\ssum_{k=0}^\infty 
   \E[x_k \trans{x_k}] Q)} = \tr{HQ}$, where $H = \lyaphsa^{-1}(I) \sgeq 0$ and
   apply \cite[Proposition 2.1]{Recht2010} to show that 
   $I = \argmax_{\nrm{Q} \leq 1} \tr HQ$. Therefore $\nrm{\lyapcl^{-1}}\leq
   \nrm{\lyapcl^{-1}(I)}$. 
\end{proof}
\archiv{
The proofs of the Lemmas in this section are deferred to the technical report
associated with this paper \cite{Coppens2020}. 
}{}

\section{RICCATI PERTURBATION} \label{sec:ric}
In this section we study the stochastic Riccati equation with perturbed
parameters. The goal is to bound how much such perturbations affect the
solutions, thereby proving Theorem~\refta{}. To do so we will use the methodology applied in \cite{Konstantinov1993},
\cite{Sun1998} to the deterministic case. 
The main proof is stated at the end of the section, for which we state
the main component first. This is a reformulation of
the perturbed Riccati equation as a fixed-point equation.

More specifically let $P^\star$ be the solution of $\ric(P^\star, \Sig)=0$ and $\Sigd$ 
selected such that $\Sig + \Sigd \in \Symmpsd^{n_w}$.  
Then a $\Delta P \in \Symm^{n_x}$ 
is a solution of $\ric(P^\star + \Delta P, \Sig + \Sigd) = 0$ \iff{} it is a 
solution to the following fixed-point equation:
\begin{equation} \label{eq:ricfp}
    \ricfp(\Delta P) \dfn \lyapcl^{-1}(\ricz(\Delta P) + \ricd(\Delta P))
    = \Delta P,
\end{equation}
with $\lyapcl$ the Lyapunov operator for the optimal closed-loop system ---
which is invertible since the closed-loop system is \mss{} \cite{Morozan1983} ---
and where
\begin{align*}
    \ricz(\Delta P) &\dfn \ric(P^\star + \Delta P, \Sig) - \lyapcl(\Delta P) \\
    \ricd(\Delta P) &\dfn \ric(P^\star + \Delta P, \Sig + \Sigd) - 
        \ric(P^\star + \Delta P, \Sig).
\end{align*}
Using the constants in Table~\ref{tab:constants}, Lemma~\ref{lem:ricfp} 
then describes two essential properties of $\ricfp$.\archiv{}{{ }The proof is 
deferred to Appendix~\ref{sec:ric:a}.} 
\begin{lemma} \label{lem:ricfp}
    Let $\ricfp$ be defined as in \eqref{eq:ricfp} and $\mathcal{D} \dfn \big\{ 
        \Delta P \in \Symm^{n_x} \,|\, \nrm{\Delta P} \leq \epsP \leq \mu_{P}^\star, \Delta P + P^\star \sgeq 0 \big\}$.
    For every $\Delta P \in \mathcal{D}_P$
    \begin{theoremcontents}
        \item the spectral norm of $\ricfp(\Delta P)$ is bounded as:
        \begin{align} 
            \nrm{\ricfp(\Delta P)} &\leq h(\epsP, \sigm) \nonumber \\ &= 
            \etaLyap^{-1}(\sigm(\etaA + \kappaA 
            \epsP) + \kappaS \epsP^2),  \label{eq:ricfpb}
        \end{align}
        \item the matrix $\ricfp(\Delta P)$ is symmetric.
    \end{theoremcontents}
\end{lemma}
\archiv{
The proof of Lemma~\ref{lem:ricfp} is deferred to the
technical report associated with this paper \cite{Coppens2020}. It consist
of a set of lengthy algebraic derivations based on the matrix inversion lemma
and other standard techniques.
}{}

\begin{proof}[Proof of Theorem~\refta{}]
We can now complete proof of Theorem~\refta{}. To do so first note that 
$h(\epsP, \sigm) = \epsP$ is a quadratic equation. It is easy to check that 
\eqref{eq:ricpcond} is a necessary and sufficient condition for the existence 
of a positive solution, which is given by \eqref{eq:ricp}. By assumption we then 
also have that $\epsP \leq \mu_{P}^\star \dfn \min \sigma(P^\star)$.

Under these conditions we can verify three properties of the mapping $\Phi$:
\begin{inlinelist}
    \item it preserves symmetry,
    \item $\nrm{\Delta P} \leq \epsP$ implies $\nrm{\Phi(\Delta P)} \leq
    \epsP$,
    \item $P^\star + \Delta P \sgeq 0$ implies $P^\star + \Phi(\Delta P) \sgeq 0$. 
\end{inlinelist}
Property (i) directly follows from Lemma~\ref{lem:ricfp}.
From \eqref{eq:ricfpb} we also know that for every 
$\Delta P \in \mathcal{D}_P$ we have $\nrm{\Phi(\Delta P)} \leq 
h(\epsP, \sigm) = \epsP$, which implies property (ii).  
Since $\nrm{\ricfp(\Delta P)} \leq \epsP \leq \mu_{P}^\star$
property (iii) holds as well. Therefore $\ricfp(\mathcal{D}_P)
\subseteq \mathcal{D}_P$ and we can apply the \lschauderf{} 
\cite[Corollary 17.56]{Border2006},
which proves $\Delta P \in \mathcal{D}_P$ and therefore Theorem~\refta{}. 
\end{proof}

\section{CONTROLLER PERTURBATION} \label{sec:ctr}
In this section we derive a bound on
$\nrm{K^\star - \hat{K}}$, thereby proving Theorem~\reftb{}.
We state the proof at the end of the section, but first introduce some
of the components.

We will use a result from convex optimization, \cite[Lemma 1]{Mania2019},
which we can apply since both the {nominal} as well as the
{empirical} controllers are optima of the following cost functions:
\begin{align}
    f^*(u) &= \trans{u}Ru + 
(\Am x + \Bm u)^\top (\Sigma \otimes P^\star) (\Am x + \Bm u) \label{eq:cost} \\
\hat{f}(u) &= \trans{u}Ru + 
(\Am x + \Bm u)^\top (\Sigma \otimes \hat{P}) (\Am x + \Bm u). 
\label{eq:cost_est}
\end{align}
Both functions are strongly convex with $\mu_R = \min \sigma(R)$. 
We will then need a bound for
$\nrm{\nabla f^\star(u) - \nabla \hat{f}(u)}$:
\begin{lemma} \label{lem:gradient_bound}
Let $f^\star(u)$ and $\hat{f}(u)$ defined respectively as in \eqref{eq:cost}
and \eqref{eq:cost_est}. Then the difference between their gradients is bounded
as:
\begin{align} 
    \nrm{\nabla f^\star(u) &- \nabla \hat{f}(u)} \leq 
    ((1+\sigm)\epsP\kappaB + \sigm\etaB)\nrm{u} \nonumber \\ &+
    ((1+\sigm)\epsP\kappaAB + \sigm\etaAB)\nrm{x}.\label{eq:gradient_bound}
\end{align}
\end{lemma}
\begin{proof}
The gradients are given by $\nabla f^\star(u) = (\Bt(\Sig \otimes P^\star)\Bm + R)u
+ \Bt(\Sig \otimes P^\star) \Am x$ and $\nabla \hat{f}(u) = (\Bt(\Sigh \otimes 
\hat{P})\Bm + R)u + \Bt(\Sigh \otimes \hat{P})\Am x$. The difference between
the first terms of the gradients can be bounded by using Lemma~\ref{lem:bndineq}
and Lemma~\ref{lem:bnd2}. More specifically we have
\begin{align*}
    & (\Bt(\Sigh \otimes \hat{P})\Bm + R)u - (\Bt(\Sig \otimes P^\star)\Bm + R)u
    \\ &  
    = \Bt(\Sig \otimes \Delta P + \Sigd \otimes (P^\star + \Delta P))\Bm u.
\end{align*}
We can remove the dependency on $\Sigd$ and $\Delta P$ by applying
Lemma~\ref{lem:bndineq} and Lemma~\ref{lem:bnd2} respectively, resulting in:
\begin{align*}
    & \Bt(\Sig \otimes \Delta P + \Sigd \otimes (P^\star + \Delta P))\Bm 
    \\ &  
    \sleq \Bt((1+\sigm)\epsP(\Sig \otimes I) + \sigm(\Sig \otimes P^\star))\Bm.
\end{align*}
Since Lemma~~\ref{lem:bnd2} and Lemma~\ref{lem:bndineq} are not applicable for
the difference between the second term of the gradients we instead produce the
following bound:
\begin{align*}
    &\nrm{\Bt(\Sigh \otimes \hat{P})\Am x - \Bt(\Sig \otimes P^\star)\Am x}
    \\ &
    = \nrm{\Bt((\Sig \otimes \Delta P) + \Sigd \otimes (P + \Delta P))\Am}
    \nrm{x} \\
    &\leq \nrm{\Bm} \nrm{\Am} \nrm{\Sig}\nrm{\Delta P} \\ 
    &\qquad +\nrm{\Bm} \nrm{\Am} 
    \nrm{\Sigd} (\nrm{P^\star} + \nrm{\Delta P})) \nrm{x} \\
    &=  ((1+\sigm)\epsP \nrm{\Bm} \nrm{\Am} \nrm{\Sig} \\
    &\qquad + \sigm
    \nrm{\Bm} \nrm{\Am} \nrm{\Sig} \nrm{P^\star})\nrm{x},
\end{align*}
where we used \eqref{eq:confsig} for $\nrm{\Delta \Sig} \leq \sigm \nrm{\Sig}$.
Using the definitions of $\kappaB$, $\kappaAB$, $\etaB$ and $\etaAB$ we 
get \eqref{eq:gradient_bound}.
\end{proof}

We are now ready to bound $\nrm{\hat{K} - K^\star}$. 
\begin{proof}[Proof of Theorem~\reftb{}]
    \allowdisplaybreaks
    Let $x$ be any vector with $\nrm{x} = 1$. Then we have
    \begin{align}
        \nrm{(\hat{K} &- K^\star)x} \leq \nrm{\hat{u} - u^\star} \nonumber \\
        &\leq
        {\mu_R}^{-1} ((1+\sigm)\epsP\kappaB + \sigm\etaB)\nrm{K^\star x} \nonumber \\
        & \quad + {\mu_R}^{-1} ((1+\sigm)\epsP\kappaAB + \sigm\etaAB)\nrm{x}, \label{eq:ineqkdiff}
    \end{align}
    where we used $\hat{K} x = \hat{u}$ and $K^\star x = u^\star$ for the first
    inequality and we combined \cite[Lemma 1]{Mania2019} and Lemma~\ref{lem:gradient_bound}
    for the second inequality. Let $x^\star = \argmax_{\nrm{x} = 1} \nrm{(\hat{K} - K^\star)x}$,
    for which $\nrm{(\hat{K} - K^\star)x^\star} = \nrm{\hat{K} - K^\star}$.
    Then \eqref{eq:ineqkdiff} also holds for $x^\star$. If we also use
    $\nrm{K^\star x^\star} \leq \nrm{K^\star} \nrm{x^\star} = \nrm{K^\star}$,
    then we have proven Theorem~\reftb{}. 
\end{proof}

\section{SUBOPTIMALITY} \label{sec:reg}

This section is dedicated to the proof of the main result of this paper. 
More specifically we derive a bound for the suboptimality of the
empirical controller compared to the nominal one, 
given in \eqref{eq:regret} as a part of Theorem~\reftc{}. The proof of which
is stated at the end of this section. 

We first introduce the main component,
which is a perturbation bound
on a matrix operator, similar to the one derived in Section~\ref{sec:ric}.
More specifically we will study the adjoint Lyapunov operator
$\lyaph: \Re^{n_x \times n_x} \times \Re^{n_u \times n_x} \rightarrow 
\Re^{n_x \times n_x}$
for the closed-loop system $x_{k+1} = (A(w) + B(w)\hat{K})x_k$ given by
$\lyaph(X, K) \dfn X - 
\trans{{(\Am + \Bm K)^\#}}(\Sig \otimes X)(\Am + \Bm K)^\#$.
    In the remainder of this 
    section we will omit the second argument of $\lyaph$ and use a subscript $\star$, when $K^\star$ is implied
    (\ie{} $\lyaph(X, K^\star) = \lyaphsa(X)$).
    This corresponds with the definition in Table~\ref{tab:constants}. 

We can then state the following lemma
\begin{lemma} \label{lem:lyappb}
    Let $\covarinf$ and $\covarinfh = \covarinfh + \covarinfd$ denote the 
    solution
    to $\lyaphsa(\covarinf) = \xinit \trans{\xinit}$ and 
    $\lyaph(\covarinfh, \hat{K}) = \xinit \trans{\xinit}$
    respectively.
    Then $\covarinfd$ is also the solution of the following fixed-point equation:
    \begin{align}
        \covarinfd &= \lyapfp(\covarinfd) \dfn {\lyaphsa}^{-1}\big( \label{eq:lyapfp} \\
            & \trans{{{\Lm}^\#}} (\Sig \otimes (\covarinf + \covarinfd))
            (\Bm \Delta K)^\# 
            \nonumber \\
            & + 
            \trans{{(\Bm \Delta K)^\#}} 
            (\Sig \otimes (\covarinf + \covarinfd)) {\Lm}^\# 
            \nonumber \\
            & +
            \trans{{(\Bm \Delta K)^\#}} 
            (\Sig \otimes (\covarinf + \covarinfd)) 
            (\Bm \Delta K)^\# \nonumber
        \big),
    \end{align}
    with $\Delta K = \hat{K} - K^\star$. 
    The following bounds then hold
    \begin{align}
        &\nrm{\lyapfp(\covarinfd)} \leq g(\nrm{\covarinfd}) \nonumber \\
        & \dfn \etaLyaph^{-1} 
                (2 \kappaLB \epsK + \kappaBc \epsK^2) 
                (\nrm{\covarinf} + \nrm{\covarinfd}), 
                \label{eq:lyapfpa}\\
        &\nrm{\lyapfp(\covarinfd) - \lyapfp(\covarinfd')} \nonumber \\
        & \leq \etaLyaph^{-1} 
        (2 \kappaLB \epsK + \kappaBc \epsK^2) 
        \nrm{\covarinfd - \covarinfd'},\label{eq:lyapfpbb}
    \end{align}
    with $\kappaLB$, $\kappaBc$ and $\etaLyaph$ defined as in 
    Table~\ref{tab:constants}.
\end{lemma}
\begin{proof}
We can use basic algebra and $\lyaphsa(\covarinf) = \xinit \trans{\xinit}$ 
to rewrite the perturbed Lyapunov equation 
$\lyaph(\covarinf + \covarinfd, K^\star + \Delta K) = \xinit \trans{\xinit}$
as in \eqref{eq:lyapfp}. 
The bounds are then derived by noting that 
$\nrm{\trans{{(\Bm \Delta K)^\#}} (\Sig \otimes \covarinf) (\Bm \Delta K)^\#}=
\nrm{\trans{{\Bm^\#}} (\Sig \otimes \Delta K \covarinf \trans{\Delta K}) \Bm^\#}
\leq \kappaBc \nrm{\Delta K}^2 \nrm{\covarinf}$. We also know 
$\nrm{\trans{{(\Bm \Delta K)^\#}} (\Sig \otimes \covarinf) {\Lm}^\#} \leq
\nrm{\Lm^\#} \nrm{\Bm^\#}  \nrm{\Sig}\nrm{\Delta K} \nrm{\covarinf} = 
\etaLBc \nrm{\Delta K} \nrm{\covarinf}$. Using the same tricks for the final
term in \eqref{eq:lyapfp} and the definition of $\etaLyaph$ allows us
to prove \eqref{eq:lyapfpa}. The same tricks also produce \eqref{eq:lyapfpbb}
where we use $A \otimes B + A \otimes C = A \otimes (B+C)$. 
\end{proof}

Similarly to how Lemma~\ref{lem:ricfp} was used to bound the Riccati perturbation,
we can bound $\nrm{\covarinfd}$.
\begin{lemma} \label{lem:lyapfpb}
    Suppose $\etaLyaph^{-1} 
    (2 \kappaAB \nrm{\Delta K} + \kappaBc \nrm{\Delta K}^2) < 1$ then
    we can bound $\nrm{\covarinfd}$ as
    \begin{equation} \label{eq:lyapfpb}
        \nrm{\covarinfd} \leq 
        \frac{2 \etaLBc \nrm{\Delta K} 
        + \kappaBc \nrm{\Delta K}^2}
        {\etaLyaph - 2 \etaLBc \nrm{\Delta K} 
        - \kappaBc \nrm{\Delta K}^2}\nrm{\covarinf} 
    \end{equation}
    and $\hat{K}$ renders the true system \mss{}.
\end{lemma}
\begin{proof}
    It is easy to verify that the solution to $\epsX = g(\epsX)$ 
    is given by the right-hand side of \eqref{eq:lyapfpb}, with $g$ as
    defined in Lemma~\ref{lem:lyappb}.
    Let
    $\mathcal{D}_X = \big \{ \covarinfd \in \Symm^{n_x} \mid \nrm{\covarinfd} \leq \epsX,
    \covarinf + \covarinfd \sgeq 0 \big \}$. Then, due to 
    $\etaLyaph^{-1} 
    (2 \kappaAB \nrm{\Delta K} + \kappaBc \nrm{\Delta K}^2) < 1$ and Lemma~\ref{lem:lyappb},
    the operator $\lyapfp$ is a contraction on $\mathcal{D}_X$.
    Invoking the \lbanachf{} \cite[Theorem 3.48]{Border2006}
    then guarantees that $\mathcal{D}_X$ contains a fixed-point of $\lyapfp$. 
    Therefore $\hat{K}$ stabilizes the true system since 
    $\nrm{\covarinf + \covarinfd} \leq \nrm{\covarinf} + \epsX$ 
    is finite, implying \mss{} and 
    we have $\nrm{\covarinfd} \leq \epsX$.
\end{proof}

We are now ready to prove the suboptimality bound
\begin{proof}[Proof of Theorem~\reftc{}]
    We start by using \cite[Lemma 3.5]{Gravell2019}, which
    states:
    \begin{equation} \label{eq:costdifference}
        \costact{\xinit} - \cost{\xinit} = 
        \tr \left[ \trans{\Delta K} (R + \G(P^\star)) \Delta K \covarinfh  \right],
    \end{equation}
    where $\hat{K} = K^\star + \Delta K$. 
    Let $\etaBR = \nrm{\G(P^\star) + R}$ and $\bar{n} = \min \{ n_x, n_u \}$.
    Then consider two matrices $A, B \in \Symm^n_x$ and let $\sigma_i$ 
    denote the $i$'th smallest eigenvalue of a matrix. Then we can show
    $\tr[AB] \leq \ssum_{i=1}^n \sigma_i(A) \sigma_i(B) \leq \tr[A] \nrm{B}$,
    where we used \emph{von Neumann's trace theorem} \cite[Theorem 7.4.1.1]{Horn2012} for the first inequality and the
    definition of the spectral norm and $\tr[A] = \ssum_{i=1}^n \sigma_i(A)$ 
    for the second. By repeatedly applying this property, we can show
    $\tr \big[ \trans{\Delta K} (R + \G(P^\star)) \Delta K \covarinfh  \big]
    \leq \tr \big[ \trans{\Delta K} \Delta K \big] \etaBR\nrm{\covarinfh}$.
    Since $\tr \big[ \trans{\Delta K} \Delta K \big] = \nrm{\Delta K}_F^2
    \leq  \bar{n} \nrm{\Delta K}^2$, we have:
    \begin{equation} \label{eq:costdifferencebound}
        \costact{\xinit} - \cost{\xinit} \leq \nm \epsK^2 \etaBR
        \nrm{\covarinfh}.
    \end{equation}
    The value of $\epsK$ is given as the right-hand side of \eqref{eq:ctrp} in
    Theorem~\reftc{}, leaving only 
    the derivation of a bound for $\covarinfh$. Note that, by definition of
    $\etaLyaph$ and since $\covarinf = \lyaphsa^{-1}(x_0 \trans{x_0})$, we have $\nrm{\covarinf} \leq \etaLyaph^{-1} \nrm{\xinit}^2$.
    Hence using Lemma~\ref{lem:lyapfpb} --- which is applicable
    due to \eqref{eq:stab} in Theorem~\reftc{} --- we can prove
    \begin{align*}
        &\nrm{\covarinf + \covarinfh} \leq \nrm{\covarinf} + \nrm{\covarinfh} \\
        &\leq \frac{1}{\etaLyaph}\left[1 +  \frac{2 \etaLBc \nrm{\Delta K} 
        + \kappaBc \nrm{\Delta K}^2}
        {\etaLyaph - 2 \etaLBc \nrm{\Delta K} 
        - \kappaBc \nrm{\Delta K}^2}\right]\nrm{\xinit}^2.
    \end{align*} 
    Substituting this into \eqref{eq:costdifference} results in \eqref{eq:regret}.
\end{proof}


\section{CONCLUSIONS AND FUTURE WORKS} \label{sec:con}
This paper studied the sample complexity of LQR applied to systems with
multiplicative noise. Overall we provided three types of sample complexities
in Theorem~\refta{}-\reftc{}.

The first is given in Theorem~\refta{},
which produces a bound on the amount of samples required to make the 
resulting problem stabilizable and the Riccati perturbation finite. 

The second sample-complexity is the one related to stability, given in
Theorem~\reftc{}. It gives a bound on the amount of samples required before
the produced controller stabilizes the true system. 

The final sample-complexity is then related to performance. It is given in 
Corollary~\ref{cor:subopt} and states that the suboptimality decreases with 
$1/\nsample$. This is the same rate as was derived for determinstic 
certainty equivalent LQR in \cite{Mania2019}. 

In future work, we aim to extend the results to partially observed systems and to the 
{distributionally robust approach}, where
the stability complexity is absent, since it is satisfied automatically. 

\bibliography{IEEEabrv,bibliography}

\addtolength{\textheight}{-3cm}   

\appendix

\subsection{Proof of Lemma~\ref{lem:ricfp}} \label{sec:ric:a}
    The proof of Lemma~\ref{lem:ricfp} is stated in three steps. First
    a simplified expression of $\ricfp$ is derived. 
    Then the bound in \eqref{eq:ricfpb} is derived.
    Finally symmetry of $\ricfp(\Delta P)$ is verified. 

    \paragraph{Simplification of $\ricfp$} The result
    given in \eqref{eq:ricfp} follows from inspection. 
    To prove Lemma~\ref{lem:ricfp} however we require a simplified expression of
    $\ricz$ and $\ricd$. To derive these we can use the 
    matrix inversion lemma\footnote{
        Consider the matrices $X$, $Y$, $U$, $V$ of appropriate dimensions,
        with $X$ and $Y$ invertible. 
        The matrix inversion lemma states: If $(Y^{-1} + VX^{-1}U)$ is invertible
        then the inverse of $(X+UYV)$ exists and is given by
        $(X+UYV)^{-1} = X^{-1} - X^{-1}U(Y^{-1} + VX^{-1}U)^{-1}VX^{-1}$ 
        \cite{Woodbury1950}.
    }, to produce a simplified version
    of the Riccati equation, similarly to \cite{Konstantinov1993} 
    and \cite{Mania2019}:
    \begin{align} 
       \ric(P^\star, \Sig) &= P^\star - Q \label{eq:rics} \\ 
       &- \At(\oPe)(I + \Sm(\oPe))^{-1} \Am  = 0. \nonumber
    \end{align}
    %
    Where we applied the lemma to $\ric(P^*, \Sig) = P^* - Q
    - \At\oP(I - I \Bm (R+ \Bt \oP \Bm)^{-1} \Bt \oP)\Am$, with $X = I$, $Y = R^{-1}$,
    $U = \Bm$ and $V = \Bt\oP $,  with $\oP \dfn \oPe$ and $\oPd \dfn \oPde$.
    Similarly we can verify $\Lm = (I + \Sm(\oPe))^{-1} \Am$. The inverse of
    $(I + \Sm(\oPe))$ then exists due to the matrix inversion lemma,
    since $R + \Bt(\Sig \otimes P^*)\Bm$ is invertible.
    We then want to prove the following:
    \begin{equation} \label{eq:ricz}
       \ricz(\Delta P) \dfn \Lt\oPd(I + \Sm\oP+ \Sm\oPd)^{-1} \Sm\oPd\Lm,
    \end{equation}
    To simplify the derivation of \eqref{eq:ricz}, 
    let $\Ym \dfn I + \Sm \oP$ and $\Em \dfn \Sm \oPd$. As mentioned
    for \eqref{eq:rics} $\Ym$ is invertible and similarly $\Ym + \Em$ 
    is invertible due to similar arguments as before. 
    More specifically we assumed $\nrm{\Delta P} \leq \mu_{P}^\star$,
    which implies that $P + \Delta P \sgeq 0$. Therefore $R + \Bt (\oP + \oPd) \Bm$
    is invertible, implying that $\Ym + \Em$ is invertible due to the matrix
    inversion lemma. 
    Therefore we can derive the following relation:
    \begin{align}
        &\trans{\Ym}(\oP + \oPd)(\Ym + \Em)^{-1} 
        \nonumber \\
        & \qquad = (I + \oP \Sm) (\oP + \oPd)(\Ym + \Em)^{-1}
        \nonumber \\
        & \qquad = (\oP + \oPd + \oP \Sm\oP + \oP \Sm \oPd)
        (\Ym + \Em)^{-1} \nonumber \\
        & \qquad = (\oP(I + \Sm \oP + \Sm \oPd) + \oPd) \nonumber \\
        & \qquad \qquad \times (I + \Sm \oP + \Sm \oPd)^{-1} \nonumber \\
        & \qquad =  \oP  + \oPd(\Ym + \Em)^{-1}.
        \label{eq:ricza}
    \end{align}
    Using this property as well as $\Lm = \Ym^{-1} \Am$
    we can write:
    {\allowdisplaybreaks
    \begin{align*}
        \ricz(\Delta P)&=\ric(P^\star+\Delta P, \Sig) - \ric(P^\star, \Sig) 
        - \lyapcl(\Delta P)\\
        &=-\At(\oP + \oPd)(\Ym + \Em)^{-1}\Am \\
        & \qquad + \At \oP \Ym^{-1} \Am + \Lt \oPd \Lm \\
        &=-\Lt \trans{\Ym}(\oP + \oPd)(\Ym + \Em)^{-1}\Am \\
        & \qquad + \At \oP \Ym^{-1} \Am + \Lt \oPd \Lm \\
        &\eqwref{eq:ricza} \smashoverbracket{-\Lt\oP\Am - \Lt \oPd 
        (\Ym + \Em)^{-1} \Am}{} \\
        & \qquad + \At \oP \Ym^{-1} \Am + \Lt \oPd \Lm \\
        &= \Lt[ -\oP \Ym - \oPd 
        (\Ym + \Em)^{-1} \Ym \\
        & \qquad + \trans{\Ym} \oP + \oPd ]\Lm,
    \end{align*}
    The final equation can be simplified further 
    to \eqref{eq:ricz} by noting that 
    $\trans{\Ym} \oP = \oP \Ym$ and $I - (\Ym + \Em)^{-1}Y = 
    (\Ym + \Em)^{-1}(\Ym + \Em - \Ym) = (\Ym + \Em)^{-1}\Em$.
    }

    We then derive the a simplified description of $\ricd(\Delta P)$.   
    To do so we again define similar shorthands as we used for the derivation
    of $\ricz(\Delta P)$, but now in terms of $\hat{P} = P^* + \Delta P$. More
    specifically we now have $\oPh = \oPhe$, $\oPhd = \oPhde$, $\Yh = I + \Sm \oPh$
    and $\Eh = \Sm \oPhd$. Note that $\Yh$ and $\Yh + \Eh$ are invertible
    from the same arguments as before. The final expression is then:
    \begin{equation} \label{eq:ricd}
        \ricd(\Delta P) \dfn \At\Yh^{-\top}\oPhd (\Yh + \Eh)^{-1}\Am.
    \end{equation}
    To derive \eqref{eq:ricd} we first use the matrix inversion lemma:
    \begin{align} 
        (\Yh+\Eh)^{-1} &= 
        \Yh^{-1} - \Yh^{-1} \Eh (I + \Yh^{-1} \Eh)^{-1} \Yh^{-1} \nonumber \\
        &= \Yh^{-1}(I - \Eh (\Yh + \Eh)^{-1}). \label{eq:ricda}
    \end{align}
    Therefore we have:
    {
    \allowdisplaybreaks
    \begin{align*}
        \ricd&(\Delta P) = \ric(\hat{P}, \Sig) - \ric(\hat{P}, \Sig + \Sigd) \\
        &= \At(\oPh + \oPhd)(\Yh + \Eh)^{-1}\Am - \At \oPh \Yh^{-1} \Am \\
        &= \At[\oPhd(\Yh + \Eh)^{-1} \\
        & \qquad + \oPh(\Yh + \Eh)^{-1} - \oPh \Yh^{-1} ]\Am \\
        &\eqwref{eq:ricda}  
        \At[\oPhd(\Yh + \Eh)^{-1} \\
        & \qquad + \oPh\smashoverbracket{\Yh^{-1}
        (I - \Eh (\Yh + \Eh)^{-1})}{} - \oPh \Yh^{-1} ]\Am \\
        &= \At[\oPhd - \oPh \Yh^{-1} \Eh ](\Yh + \Eh)^{-1}\Am \\
        &= \At[I - \oPh \Yh^{-1} \Sm ]\oPhd (\Yh + \Eh)^{-1}\Am,
    \end{align*}
    where we used the definition of $\Eh$ for the final equality. We can further
    simplify the quantity between square brackets by applying 
    the matrix inversion lemma once more:
    \begin{equation*}
        I - \oPh \Yh^{-1} \Sm = I - \oPh (I + \Sm \oPh)^{-1} \Sm = \Yh^{-\top}.
    \end{equation*}
    Therefore we have proven \eqref{eq:ricd}.
    Using the simplification of $\ricz$ and $\ricd$ in \eqref{eq:ricz} and
    \eqref{eq:ricd} respectively we can move on to proving the bound in
    \eqref{eq:ricfpb}.

    \paragraph{Bound on $\ricfp(\Delta P)$}
    By the definition of $\ricfp$ and
    by the definition of the operator norm of a linear matrix
    function we can write:
    \begin{equation*}
        \nrm{\ricfp(\Delta P)} \leq \nrm{\ricf^{-1}} (\nrm{\ricz(\Delta P)}
        + \nrm{\ricd(\Delta P)}).
    \end{equation*}
    Therefore we will, in order, focus on bounding 
    $\nrm{\ricz(\Delta P)}$, $\nrm{\ricd(\Delta P)}$ and $\nrm{\ricf^{-1}}$.
    To do so we will need the following Lemma, which is a generalisation of
    \cite[Lemma 7]{Mania2019}:
    \begin{lemma} \label{lem:ineq}
        Let $M, N \in \Symm^n$ and $(I + MN)$ invertible. 
        Then
        \begin{theoremcontents}
            \item \hypertarget{lem:ineq:a}{} $(I+MN)^{-1}M = M(I+NM)^{-1}$, 
            \item \hypertarget{lem:ineq:b}{} $MN(N^\pinv + M)NM \sgeq 0 \Leftrightarrow (I+MN)^{-1}M \sleq M$,
            \item \hypertarget{lem:ineq:c}{} $MN(2N^\pinv + M)NM \sgeq 0 \\$
            $\qquad \Leftrightarrow (I + MN)^{-1}M(I+NM)^{-1} \sleq M$.
        \end{theoremcontents}
    \end{lemma}
    \begin{proof}
        We first prove \emph{\textbf{i)}}. To do so we
        apply the matrix inversion lemma:
        \begin{align*}
            (I &+ MIN)^{-1} M 
            = (I - M(I+ NM)^{-1}N)M \\
            &= M(I+NM)^{-1} (I+NM) - M(I+NM)^{-1} NM\\
            &= M(I+NM)^{-1} = ((I + MN)^{-1} M)^\top.
        \end{align*}
        Therefore \emph{\textbf{i)}} is verified. 
        To prove \emph{\textbf{ii)}} we write the matrix inequality as:
        \begin{equation} \label{eq:ineq1}
            \trans{x} M x \geq \trans{x} (I + MN)^{-1} M x, \quad \forall x.
        \end{equation}  
        Let $y = (I + NM)^{-1}x$, then we can substitute it into \eqref{eq:ineq1}
        to get an equivalent condition:
        \begin{align}
            &\trans{y} (I+MN) M (I + NM) y \nonumber \\
            & \qquad \geq \trans{y}  M (I+NM) y, \, \forall y \nonumber \\
            \Leftrightarrow \, & \trans{y} (MNM + MNMNM) y \geq 0, \quad \forall y. \label{eq:ineq2} 
        \end{align}
        Using the fact that a pseudoinverse is a weak inverse (\ie{} $N N^\pinv N = N$),
        we can then show that \eqref{eq:ineq2} is equivalent to
        \begin{equation*}
            \trans{y} MN(N^\pinv + M)NM y \geq 0, \quad \forall y,
        \end{equation*}
        showing the required result in \emph{\textbf{i)}}. 
        The proof for \emph{\textbf{iii)}} follows a similar procedure. 
        \peter{It would look as follows. 
        We want to find an equivalent expression for:
        \begin{equation} \label{eq:ineq5}
            \trans{x} M x \geq \trans{x} (I + MN)^{-1} M (I + NM)^{-1}x, \quad \forall x.
        \end{equation}
        Let $y = (I+NM)^{-1}x$, then we can substitute it into \eqref{eq:ineq5}
        to get an equivalent condition:
        \begin{align}
            &\trans{y} (I+MN) M (I + NM) y \nonumber \\
            & \qquad \geq \trans{y} M  y, \, \forall y \nonumber \\
            \Leftrightarrow \, & \trans{y} (2MNM + MNMNM) y \geq 0, \quad \forall y. \label{eq:ineq6} 
        \end{align}
        Using the fact that a pseudoinverse is a weak inverse (\ie{} $N N^\pinv N = N$),
        we can then show that \eqref{eq:ineq6} is equivalent to
        \begin{equation*}
            \trans{y} MN(2N^\pinv + M)NM y \geq 0, \quad \forall y,
        \end{equation*}
        showing the required result in \emph{\textbf{iii)}}.
        }
    \end{proof}
    
    To bound $\nrm{\ricz(\Delta P)}$ we first employ Lemma~\reflemb{} and the
    fact that $\Sm \sgeq 0$ and $\oP + \oPd \sgeq 0$ to state:
    \begin{equation*}
        \oPd(I+\Sm\oP+\Sm\oPd)^{-1}\Sm\oPd \sleq \oPd \Sm \oPd.
    \end{equation*}
    Therefore we can write:
    \begin{equation} \label{eq:riczb}
        \nrm{\ricz(\Delta P)} \leq \nrm{\Lt \oPd \Sm \oPd \Lm} \leq 
        \kappaS \nrm{\Delta P}^2,
    \end{equation}
    where the final inequality follows from the definition of $\kappaS
    \dfn \nrm{\Lm}^2 \nrm{\Sm} \nrm{\Sig}^2$. 
    
    {\allowdisplaybreaks
    We can apply a similar procedure to find a bound for $\nrm{\ricd(\Delta P)}$,
    the simplification however is slightly more involved. First consider the
    following:
    \begin{align*} 
        \Yh^{-\top}&\oPhd (\Yh + \Eh)^{-1} = 
        \Yh^{-\top} \oPhd (I + \Yh^{-1} \Eh)^{-1} \Yh^{-1} \\
        &=\Yh^{-\top} \oPhd (I + (I+\Sm \oPh)^{-1} \Sm \oPhd)^{-1}  \Yh^{-1} \\
        &\sleq \Yh^{-\top} \oPhd \Yh^{-1},
    \end{align*}
    where the inequality follows from Lemma~\reflema{} 
    followed by Lemma~\reflemb{}. To prove that this
    is allowed, note that $N = (I+\Sm \oPh)^{-1} \Sm$ and $M = \oPhd$. 
    We need to prove that $N \in \Symm^{n_x n_w}$ and
    that $MN(N^\pinv + M)NM \sgeq 0$. Symmetry follows from Lemma~\reflema{}. 
    The matrix inequality is then verified by noting that 
    $\Sm(N^\pinv + M)\Sm = \Sm\Sm^\pinv \Sm + \Sm \Sm^\pinv \Sm \oPhd \Sm + 
    \Sm\oPhd \Sm = \Sm(\Sm^\pinv + \oPh + \oPhd)\Sm \sgeq 0$,
    where we used $\Sm \Sm^\pinv \Sm  = \Sm$.
    The next step involves getting rid of $\Sigd$, for which we will invoke
    Lemma~\ref{lem:bndineq}:
    \begin{align*}
        \ricd(\Delta P) &\sleq \At \Yh^{-\top} (\Sigd \otimes \hat{P}) \Yh^{-1} \Am 
        \\ &\sleq \sigm \At \Yh^{-\top} (\Sig \otimes \hat{P}) \Yh^{-1} \Am,
    \end{align*}
    To produce the final bound we still need to get rid of the inverses of $\Yh$. 
    More specifically we can prove that
    \begin{equation} \label{eq:ineqricdother}
        \sigm \At \Yh^{-\top} (\Sig \otimes \hat{P}) \Yh^{-1} \Am \sleq 
        \sigm \At (\Sig \otimes \hat{P}) \Am
    \end{equation}
    by using Lemma~\reflemc{} and $\Sm \sgeq 0$ and $\oPh \sgeq 0$. Remember how 
    $\F(P) = \At (\Sig \otimes P) \Am$, hence we can write:
    \begin{align} 
        \nrm{\ricd(\Delta P)} &\leq \sigm \nrm{\F(P^\star + \Delta P)} \nonumber \\
        & \leq \sigm (
        \nrm{\F(P^\star)} + \nrm{\F(\Delta P)}) \nonumber \\
        & = \sigm (\etaA +\kappaA\nrm{\Delta P}),
        \label{eq:ricdb}
    \end{align}
    where we used the definitions of $\etaA$ and $\kappaA$ for the final equality. 
    Therefore, by combining \eqref{eq:riczb}, \eqref{eq:ricdb} and the
    definitions of $\etaL$ and $\sigm$, we have proven that, $\forall \Delta P \in 
    \mathcal{D}_P$:
    \begin{align*}
        \nrm{\ricfp(\Delta P)} &\leq h(\nrm{\Delta P}, \sigm) \\ 
        &= \etaLyap^{-1}(\sigm(\etaA + \kappaA 
        \nrm{\Delta P}) + \kappaS \nrm{\Delta P}^2).
    \end{align*}
    This is easily extended to \eqref{eq:ricfpb} by noting that $h(\nrm{\Delta P}, 
    \sigm)$ 
    increases monotonically for increasing values of $\nrm{\Delta P}$. 
    Therefore for all $\nrm{\Delta P} \leq \epsP$ we have that 
    $h(\nrm{\Delta P}, \sigm) \leq h(\epsP, \sigm)$, which implies
    \eqref{eq:ricfpb}.
    }
    \paragraph{Symmetry of $\ricfp(\Delta P)$}  
    The final step is to show that $\ricfp(\Delta P)$ is symmetric whenever
    $\Delta P$ is symmetric. To see this first note that $\ricz(\Delta P)$ 
    and $\ricd(\Delta P)$ are symmetric, since it follows from inspection that $\lyapcl$ and
    $\ric$ preserve the symmetry of the input. 
    Therefore if we prove that $\lyapcl^{-1}$
    preserves symmetry then $\ricfp$ does so as well. To do so note that
    $\lyapcl(P) = P - \F(P)$. It is then easy to verify that $P 
    = \ssum_{k=0}^{\infty} \F^k(Q)$, where $\F^k(Q)$ implies that $\F$ is 
    applied $k$ times, solves $\lyapcl(P) = Q$. 
    After all $P - \F(P) = \ssum_{k=0}^{\infty} \F^k(Q) - \ssum_{k=1}^{\infty} \F^k(Q) = Q$.
    Since $\F$ clearly preserves
    symmetry, $\lyapcl^{-1}$ does so as well. \hfill{} \qedsymbol

\end{document}